\documentclass[aps,amsmath,a4paper,superscriptaddress,notitlepage,twoside]{revtex4-1}
\usepackage{geometry}
\usepackage{amssymb}
\usepackage{amsmath}
\usepackage{graphicx}
\usepackage{amsfonts}
\usepackage{hyperref}
\usepackage{color}
\usepackage{epsfig}
\usepackage{bbm}
\usepackage{color,enumerate,anysize,amsfonts,amsthm,mathrsfs,url}

\newcommand{\ket}[1]{\left| #1 \right>}
\newcommand{\bra}[1]{\left< #1 \right|}
\newcommand{\proj}[1]{\left| #1 \right> \! \left< #1 \right|}
\newcommand{\braket}[2]{\left< #1 | #2 \right>}

\newcommand{\hilb}[1]{\mathscr{H}_{#1}}
\newcommand{\basis}[1]{\mathscr{B}_{#1}}
\newcommand{\real}[0]{\mathbb{R}}

\newcommand{\cpx}[0]{\mathbb{C}}

\newcommand{\rk}{\operatorname{rank}}
\newcommand{\vecspan}{\operatorname{span}}
\newcommand{\supp}{\operatorname{supp}}
\newcommand{\trace}{\operatorname{Tr}}
\newcommand{\tr}{\operatorname{Tr}}
\newcommand{\1}{\openone}

\theoremstyle{plain} \newtheorem{thm}{Theorem}
\theoremstyle{plain} \newtheorem*{thm'}{Theorem 7$'$}
\theoremstyle{plain} 
\theoremstyle{definition} \newtheorem*{rmk}{Remark}
\theoremstyle{definition} 
\theoremstyle{plain} \newtheorem{lemma}[thm]{Lemma}
\theoremstyle{definition} 
\theoremstyle{plain} 
\theoremstyle{plain} \newtheorem{hyp}{Hypothesis}

\begin{document}

\title{Inequalities for the Ranks of Quantum States}

\author{Josh Cadney}
\affiliation{School of Mathematics, University of Bristol, Bristol BS8 1TW, United Kingdom}

\author{Marcus Huber}
\affiliation{F\'isica Te\`{o}rica: Informaci\'o i Fenomens Qu\`{a}ntics, Universitat Aut\'{o}noma de Barcelona, ES-08193 Bellaterra (Barcelona), Spain}
\affiliation{School of Mathematics, University of Bristol, Bristol BS8 1TW, United Kingdom}
\affiliation{ICFO-Institut de Ci\`encies Fot\`oniques, 08860 Castelldefels, Barcelona, Spain}

\author{Noah Linden}
\affiliation{School of Mathematics, University of Bristol, Bristol BS8 1TW, United Kingdom}

\author{Andreas Winter}
\affiliation{ICREA -- Instituci\'{o} Catalana de Recerca i Estudis Avan\c{c}ats, Pg.~Lluis Companys 23, ES-08010 Barcelona, Spain}
\affiliation{F\'isica Te\`{o}rica: Informaci\'o i Fenomens Qu\`{a}ntics, Universitat Aut\'{o}noma de Barcelona, ES-08193 Bellaterra (Barcelona), Spain}
\affiliation{School of Mathematics, University of Bristol, Bristol BS8 1TW, United Kingdom}

\date{23 July 2013}

\begin{abstract}
We investigate relations between the ranks of marginals of multipartite quantum states. 
These are the Schmidt ranks across all possible bipartitions and constitute a natural 
quantification of multipartite entanglement dimensionality. We show that there exist 
inequalities constraining the possible distribution of ranks. This is analogous to the 
case of von Neumann entropy ($\alpha$-R\'enyi entropy for $\alpha=1$), where nontrivial 
inequalities constraining the distribution of entropies (such as e.g. strong subadditivity) 
are known. It was also recently discovered that all other $\alpha$-R\'enyi entropies 
for $\alpha\in(0,1)\cup(1,\infty)$ satisfy only one trivial linear inequality 
(non-negativity) and the distribution of entropies for $\alpha\in(0,1)$ is completely 
unconstrained beyond non-negativity. 
Our result resolves an important open question by showing that also the 
case of $\alpha=0$ (logarithm of the rank) is restricted by nontrivial linear relations 
and thus the cases of von Neumann entropy (i.e., $\alpha=1$) and $0$-R\'enyi entropy 
are exceptionally interesting measures of entanglement in the multipartite setting.
\end{abstract}

\maketitle

\section{Introduction}
Entanglement is ubiquitous in multi-party quantum systems; today it is recognized 
to play a fundamental role as a resource in quantum information theory. 
Therefore there have been intensive investigations into quantifying this resource 
in an operational way (cf.~\cite{Horodqe} for a recent review). Apart from its 
applications in metrology and communication it is also necessary for exponential 
speedup in quantum computation \cite{Vid03}. Many continuous measures of entanglement 
however may be polynomially small in the size of the system while still admitting 
an exponential speedup \cite{vandenNest}. In this scenario the Schmidt rank 
constitutes a natural measure of bipartite entanglement, but more importantly in 
terms of a resource theory it reveals the minimum entanglement dimensionality 
required to create a state via stochastic local operations and classical communication 
(SLOCC) \cite{LP01}. 
The Schmidt rank is the number of non-zero terms in the Schmidt decomposition of 
$\ket{\psi}_{AB}$, or more simply, the minimum number of terms in a decomposition
of $\ket{\psi}_{AB} = \sum_{i=1}^r \ket{\alpha_i}_A \ket{\beta_i}_B$ into a sum of 
product vectors. Equivalently, it is the rank of the reduced density matrix 
$\psi_A := \trace_B \ket{\psi}\!\bra{\psi}$. 

Recently, this idea has been generalized to the multipartite setting by the 
introduction of the \emph{rank vector} \cite{HdV13,HPdV}. The rank vector is a list 
of the Schmidt ranks across each possible bipartition of the state. Each 
element of this vector constitutes an SLOCC monotone and it can be used to 
reveal the dimensionality of genuine multipartite entanglement. For example, 
a tripartite state $\ket{\psi}_{ABC}$ has three bipartitions, $A:BC$, 
$B:AC$ and $AB:C$, giving three different ranks: $r_A$, $r_B$ and $r_{AB}$. 
It is then natural to ask what relations these ranks must satisfy. 
For example, the inequality
\begin{equation} \label{ranksubmult}
	r_{AB}\leq r_A \, r_B
\end{equation}
is a simple consequence of the relation
\begin{equation} \label{supporteq}
	\supp(\psi_{AB}) \subseteq \supp(\psi_A) \otimes \supp(\psi_B),
\end{equation}
where $\supp(\sigma)$ denotes the support of the density matrix $\sigma$. 
In \cite{HdV13} it is shown that for tripartite pure states this is the 
only relevant constraint on the ranks, however, it is conjectured that for 
four and more parties there are further inequalities. Finding these would 
not only reveal a great amount of structure in multipartite Hilbert spaces, 
but also give a rich set of constraints for distributing entanglement in 
multipartite systems in the flavour of generalized monogamy relations. 
For example, the inequality above makes an interesting physical statement: 
The product of the entanglement dimensionality of $A$ with $BC$ and the 
entanglement dimensionality of $B$ with $AC$, is an upper bound to the 
entanglement dimensionality of $C$ with $AB$.

An at first glance quite separate research topic is the study of entropy 
inequalities, which is important in both classical \cite{YZ98,Mat07,DFZ11} 
and quantum \cite{Pip03,LW05,CLW12} information theory. Given a multipartite 
system, how do the entropies of the different subsystems relate to each other? 
For the case of the von Neumann entropy, $S(\rho) = - \tr\rho\log\rho$,
they must satisfy the well-known strong subadditivity and weak monotonicity 
relations \cite{LR73}
\begin{align}
	S(\rho_A)+S(\rho_{ABC}) &\leq S(\rho_{AB})+S(\rho_{AC}), \label{SSA} \\
	S(\rho_A)+S(\rho_B) &\leq S(\rho_{AC})+S(\rho_{BC}).     \label{WM}
\end{align}
For four-party pure states these are the only constraints \cite{Pip03}, 
but it is a major open question to determine whether or not further inequalities 
exist for five or more parties \cite{magnificent-7,Ibi07}.

The von Neumann entropy is a special case of the family of quantum R\'enyi 
entropies. For $\alpha\in(0,1)\cup(1,\infty)$ the $\alpha$-R\'enyi entropy, 
$S_\alpha$, of a quantum state, $\rho$, is defined by
\begin{equation}
	S_\alpha(\rho) := \frac{1}{1-\alpha} \log \trace \rho^\alpha.
\end{equation}
where the logarithm is taken to base 2. The von Neumann entropy is the $\alpha=1$ R\'enyi entropy, in the sense 
that $\lim_{\alpha\to1} S_\alpha(\rho) = S(\rho)$. Recently, all (linear) 
entropy inequalities for the R\'enyi entropy were found for 
$\alpha\in(0,1)\cup(1,\infty)$ \cite{LMW12}. The surprising result was that 
the only such inequality for these entropies is non-negativity: 
$S_\alpha(\rho_A)\geq 0$. In \cite{LMW12} the special case $\alpha=0$ is 
also posed as an open question, where $S_0$ is defined by
\begin{equation}
	S_0(\rho) := \lim_{\alpha\to0} S_\alpha(\rho) = \log \rk \rho.
\end{equation}
Since the $0$-entropy is nothing but 
the logarithm of the rank of a state, it is clear that $0$-entropy 
inequalities and rank inequalities are essentially equivalent. Notice that 
inequality \eqref{ranksubmult} corresponds exactly to
\begin{equation}
	S_0(\rho_{AB}) \leq S_0(\rho_A) + S_0(\rho_B),
\end{equation}
so the $0$-entropy case is already more complex than that of 
$\alpha\neq 0,1$ as it obeys \emph{subadditivity}. It is natural to ask 
whether or not the $0$-entropy also satisfies strong subadditivity; in fact, 
it is easy to see that it does not (but cf.~also our results below):
a counterexample is given by the purification $\ket{\psi}_{ABCD}$ of
\begin{equation}
  \rho_{ABC} = \frac15\left( \proj{000} + \proj{100} +\proj{101} + \proj{110} + \proj{111} \right),
\end{equation}
which has $r_A = 2$, $r_{ABC} = 5$, and $r_{AB} = r_{AC} = 3$, hence
$S_0(A) + S_0(ABC) \not\leq S_0(AB) + S_0(AC)$.
However, are there any other further inequalities?

\medskip
In this paper we address this question, raised in both \cite{HdV13} and \cite{LMW12}, 
by proving two new inequalities for the ranks of a four-party quantum state
(section~\ref{sec:new}), after introducing a framework adapted to studying the
question of universal inequalities in section~\ref{sec:ranks}. 
Interestingly, the key ingredient in our proofs is the strong subadditivity of the
von Neumann entropy. We then construct some states with interesting ranks, which 
violate some other inequalities, including one conjectured in \cite{HdV13} 
(section~\ref{sec:rays}). 
Finally, we present another inequality, which, if true, would complete the picture 
in the four-party case. We have so far been unable to find a general proof 
of this inequality, or a counterexample for that matter, but we present a proof 
for certain special cases (section~\ref{sec:hypothesis}). We conclude 
in section~\ref{sec:conclusion}, referring several open problems to the
attention of the reader.

\section{The set of rank vectors}
\label{sec:ranks}
Given an $n$-party pure quantum state $\ket{\psi}_{X_1\ldots X_n}$, there are $2^{n-1}-1$ ways to bipartition the parties $X_1,\ldots,X_n$ in a non-trivial way. We define the rank vector of $\ket{\psi}$ to be the vector in $\mathbb{N}^{2^{n-1}-1}$ which lists the Schmidt ranks of $\ket{\psi}$ across each of these bipartitions (in some fixed order) \footnote{Note that this is slightly different from \cite{HdV13}, where the rank vector is formed by listing the ranks in decreasing order}. For example, the four-party state $\ket{\psi}_{ABCD}$ has rank vector
\begin{equation}
	r_\psi = (r_A,r_B,r_C,r_D,r_{AB},r_{AC},r_{AD})
\end{equation}
Similarly, the $0$-entropy vector of $\ket{\psi}_{ABCD}$ is
\begin{equation}
	v_\psi = (\log r_A,\log r_B,\log r_C,\log r_D,\log r_{AB},\log r_{AC},\log r_{AD})
\end{equation}

We are interested in determining which vectors are rank/$0$-entropy vectors. Therefore, we define a set $\Sigma_n$ which is the set of all $n$-party rank vectors. More precisely,
\begin{equation}
	\Sigma_n = \left\{ u\in\mathbb{N}^{2^{n-1}-1} : \exists d_1,\ldots,d_n\in\mathbb{N}, \ket{\psi}\in\cpx^{d_1}\otimes\ldots\otimes\cpx^{d_n} \text{ such that } u = r_\psi \right\}
\end{equation}
and in a completely analogous way we define $\Omega_n$ to be the set of all $n$-party $0$-entropy vectors.

For $0<\alpha\leq1$ it is known that, for any number of parties, the closure of the set of $\alpha$-entropy vectors is a convex cone (i.e. it is closed under addition and multiplication by positive real scalars) \cite{LMW12}. This means that it can be characterized only in terms of the \emph{linear} inequalities between the $\alpha$-entropy of different parts of system. Since $\Omega_n$ is a discrete set it is clearly not a cone. However, the results of \cite{HdV13} imply that $\Omega_3$ is the intersection of a cone with the set of log-integer points. In particular, let $\mathscr{C}$ be the closed cone defined by
\begin{equation}
	\mathscr{C}:= \{ (x,y,z)\in\real^3 : x,y,z\geq0, x\leq y+z, y\leq x+z, z\leq x+y \}
\end{equation}
and let $\log \mathbb{N}^3$ denote the set of log-integer points, then $\Omega_3=\mathscr{C}\cap\log\mathbb{N}^3$.
For example, if $a,b,c\in\mathbb{N}$ are such that $(\log a,\log b,\log c)\in\mathscr{C}$ with $a\leq b \leq c$ then by taking the state
\begin{equation}
	\ket{\psi}_{ABC} = \sum_{i=1}^a \sum_{j=1}^b \ket{i}_A \ket{j}_B \ket{(i-1)b+j}_C,
\end{equation}
and removing appropriate terms from the sum, we can arrive at a state with rank vector $(a,b,c)$.

It turns out that a similar statement does not hold for $\Omega_n$ with $n\geq4$.
However, it is true that $\Omega_n$ is still closed under addition, and hence also 
multiplication by positive integer scalars. For, if $\ket{\psi}$ and $\ket{\phi}$ 
have $0$-entropy vectors $v_\psi$ and $v_\phi$ respectively, then the state 
$\ket{\psi}\otimes\ket{\phi}$ has $0$-entropy vector $v_\psi+v_\phi$. This is one 
reason why we will be principally concerned with inequalities which are linear in 
the $0$-entropies of a state, and therefore geometric in the ranks.

\medskip\noindent
{\bf Remark.} The theory outlined above has a classical counterpart,
which has been studied in the past: Namely, to the joint distribution
$P_{[n]}$ of $n$ discrete random variables $X_1,\ldots,X_n$, associate for each
subset $I \subset [n]$ the cardinality $s_I$ of the support of $P_I$,
which is the marginal distribution of $X_I = (X_i:i\in I)$, so that
$\log s_I = H_0(X_I)$ is the classical $0$-R\'enyi entropy of the variables 
$X_I$.
Another way of looking at it is to observe that $s_I$ is precisely the 
size of the projection of the support of $P_{[n]}$ onto the coordinates $I$.

By introducing the $(n+1)$-party purification
\begin{equation}
  \label{classical-state}
  \ket{\psi}_{012\ldots n} 
    = \sum_{x_1\ldots x_n} \sqrt{P(x_1\ldots x_n)} \ket{x_1\ldots x_n}^0 \ket{x_1}^1 \cdots \ket{x_n}^n,
\end{equation}
we see that the $s_I$, $I\subset [n]$ are precisely the Schmidt ranks $r_I$
of this state, so that the classical case actually appears a special case
of the quantum one, for states of the form \eqref{classical-state}.

What are the inequalities satisfied by the $s_I$ of a generic distribution?
By the above observation, any constraint on the quantum ranks must
necessarily hold for the classical supports, for instance submultiplicativity
$s_{I \stackrel{.}{\cup}J} \leq s_I \, s_J$. But there is evidently more,
such as monotonicity, i.e.~$s_I \leq s_J$ for all $I \subset J$, which
does not hold for the quantum ranks $r_I$. Further, for each $J\subset [n]$
and each $1\leq k \leq |J|$ the following inequality holds \cite{SM83}
\begin{equation}
	s_J^{\binom{|J|-1}{k-1}} \leq \prod_{|I|=k} s_I,
\end{equation}
where the product is taken over all subsets $I\subset J$ of size $k$. Theorem \ref{newineq2} below
shows that the case $|J|=3$, $k=2$ is true also for quantum ranks.

\section{New rank inequalities}
\label{sec:new}
In this section we prove two new inequalities for the ranks of a four-party state.

\begin{thm} \label{newineq1}
	Let $\ket{\psi}_{ABCD}$ be a four-party quantum state. Then $r_A\leq r_{AB} \, r_{AC}$.
\end{thm}
\begin{proof}
	By \eqref{supporteq} we see that projecting $\hilb{A}$ onto $\supp(\psi_A)$ does not change $\ket{\psi}_{ABCD}$. Consequently we may assume, without loss of generality, that the single-party density matrix $\psi_A$ has full rank, i.e. $r_A=d_A$. The Schmidt decomposition across the $A:BCD$ partition then gives us $\ket{\psi}_{ABCD} = \sum_{i=1}^{d_A} \lambda_i \ket{\alpha_i}_A \otimes \ket{\beta_i}_{BCD}$, where $\{ \ket{\alpha_i} \}_{i=1}^{d_A}$ is an orthonormal basis of $\hilb{A}$, the states $\ket{\beta_i}$ are also orthonormal, and $\lambda_i>0$ for all $i$.
	
	Consider the operator 
	$\psi_A^{-\frac12} = \sum_{i=1}^{d_A} \lambda_i^{-1} \ket{\alpha_i}\!\bra{\alpha_i}$. 
	It is an invertible linear map on $\hilb{A}$. Therefore the operator 
	$X:= \psi_A^{-\frac12}\otimes \1_B \otimes \1_C \otimes \1_D$, where $\1_B$ denotes the identity on $\hilb{B}$, is a local invertible linear map on $\hilb{ABCD}$, which will not change any of the ranks. Let $\ket{\psi'}_{ABCD}:= \frac1{\sqrt{d_A}}X\ket{\psi}_{ABCD} = \sum_{i=1}^{d_A} \frac1{\sqrt{d_A}} \ket{\alpha_i}_A \otimes \ket{\beta_i}_{BCD}$.
	
	The von Neumann entropy of the reduced state $\psi_A'$ is $\log d_A=\log r_A$. The strong subadditivity of von Neumann entropy \eqref{SSA} then gives
	\begin{equation}
	\begin{split}
		\log r_A = S(\psi_A') &\leq S(\psi_{AB}') + S(\psi_{AC}') - S(\psi_{ABC}') \\
		                    	&\leq \log r_{AB} + \log r_{AC},
	\end{split}
	\end{equation}
where we have used the fact that for any density matrix $\rho$ with rank $r$, 
$S(\rho)\leq \log r$. Taking the exponential of each side of this inequality gives the result.
\end{proof}

\begin{thm} \label{newineq2}
	Let $\ket{\psi}_{ABCD}$ be a four-party quantum state. 
	Then $r_A^2 \leq r_{AB} \, r_{AC} \, r_{BC}$.
\end{thm}
\begin{proof}
	We follow the proof above, but this time we take the entropy inequality \eqref{SSA} 
	and add to it the further inequality
	\begin{equation}
		S(\psi_A') \leq S(\psi_{BC}') + S(\psi_{ABC}'),
	\end{equation}
	which is derived from \eqref{WM} by taking the $B$ system to be trivial.
	This yields
	\begin{equation}
	\begin{split}
		2 \log r_A = 2 S(\psi_A') &\leq S(\psi_{AB}') + S(\psi_{AC}') + S(\psi_{BC}') \\
			&\leq \log r_{AB} + \log r_{AC} + \log r_{BC}.
	\end{split}
	\end{equation}
	Again, taking exponentials gives the result.
\end{proof}

We have given a simple proof of each of the inequalities above. However, the 
crucial ingredient underpinning these proofs is the strong subadditivity of 
von Neumann entropy. This is a very famous, and highly non-trivial result. 
Although many proofs of strong subadditivity are known (e.g. \cite{LR73,GPW05,Rus07,BR12}), 
none provides any intuition regarding the ranks we consider. Furthermore, the
argument we have employed is limited by the fact that by a local invertible
filtering one can make only one reduced state equal to the maximally mixed state;
indeed, if it were possible to find, for given $\ket{\psi}_{ABCD}$, another
state $\ket{\psi'}_{ABCD}$ with the same Schmidt ranks for all bipartite
partitions, but such that $\psi_A'$ and $\psi_D'$ are maximally mixed on
their respective subsystems, then strong subadditivity \eqref{SSA} and weak
monotonicity \eqref{WM} would hold for $S_0$ -- and we know that already to
be false.
It would, therefore, be desirable to give direct, self-contained proofs of 
Theorems \ref{newineq1} and \ref{newineq2}. We are able to do so for the former.

\begin{proof}[Direct proof of Theorem \ref{newineq1}]
	We begin with some preliminaries. Suppose that $\ket{\psi}_{AB}$ is a bipartite quantum state, and fix orthonormal bases $\{ \ket{i}_A \}_{i=1}^{d_A}$ and $\{ \ket{i}_B \}_{i=1}^{d_B}$ of $\hilb{A}$ and $\hilb{B}$. Then, for some complex $\alpha_{ij}$, we can write
	\begin{equation}
		\ket{\psi}_{AB} = \sum_{i=1}^{d_A} \sum_{j=1}^{d_B} \alpha_{ij} \ket{ij}_{AB} = \sum_{i=1}^{d_A} \ket{i}_A \otimes \left( \sum_{j=1}^{d_B} \alpha_{ij} \ket{j}_B \right) = \sum_{i=1}^{d_A} \ket{i}_A \otimes \ket{\alpha_i}_B,
	\end{equation}
	where we define $\ket{\alpha_i}_B := \sum_{j=1}^{d_B} \alpha_{ij} \ket{j}_B$. Now let $M:\hilb{A}\to\hilb{B}$ be the linear map such that $M\ket{i}_A = \ket{\alpha_i}_B$. Then we can write
	\begin{equation} \label{matrixrep}
		\ket{\psi}_{AB} = \sum_{i=1}^{d_A} \ket{i}_A \otimes M\ket{i}_A = (\1_A\otimes M) \ket{\phi^+}_{AA},
	\end{equation}
	where $\ket{\phi^+}_{AA} := \sum_{i=1}^{d_A} \ket{i}_A \ket{i}_A$ is the 
	(unnormalized) maximally entangled state between two copies of $\hilb{A}$.
	
	We wish to calculate the reduced density matrix $\psi_B := \trace_A \psi$, which we can do as follows
	\begin{equation}
	\begin{aligned}
		\psi_B &= \sum_{j=1}^{d_A} (\bra{j}_A \otimes \1_B) \ket{\psi}\!\bra{\psi}_{AB} (\ket{j}_A \otimes \1_B) \\
			&= \sum_{i,i',j=1}^{d_A} (\bra{j}_A \otimes \1_B) (\ket{i}_A \otimes M\ket{i}_A) (\bra{i'}_A \otimes \bra{i'}_A M^\dagger) (\ket{j}_A \otimes \1_B) \\
			&= \sum_{i,i',j=1}^{d_A} \delta_{ij}\delta_{i'j} M \ket{i}\!\bra{i'}_A M^\dagger \\
			&= \sum_{j=1}^{d_A} M \ket{j}\!\bra{j}_A M^\dagger
			 = M \left( \sum_{j=1}^{d_A} \ket{j}\!\bra{j}_A \right) M^\dagger
			 = M M^\dagger.
	\end{aligned}
	\end{equation}
	From this we can conclude that the rank of $\ket{\psi}_{AB}$ across the $A:B$ partition is $\rk M$.
	
	We now return to the problem at hand, in which we have a four-party state $\ket{\psi}_{ABCD}$. Taking the Schmidt decomposition across the $AB:CD$ partition gives
	\begin{equation}
		\ket{\psi}_{ABCD} = \sum_{k=1}^{r_{AB}} \ket{\eta_k}_{AB} \otimes \ket{\theta_k}_{CD},
	\end{equation}
	where $\{ \ket{\eta_k}_{AB} \}_{k=1}^{r_{AB}}$ and $\{ \ket{\theta_k}_{CD} \}_{k=1}^{r_{AB}}$ are sets of mutually orthogonal (unnormalized) states. According to eq.~\eqref{matrixrep}, 
	write $\ket{\eta_k}_{AB} = (\1_A\otimes R_k)\ket{\phi^+}_{AA}$ and $\ket{\theta_k}_{CD} = (\1_C \otimes S_k) \ket{\phi^+}_{CC}$. Then we have
	\begin{equation} \label{RSdecomp}
		\ket{\psi}_{ABCD} = \sum_{k=1}^{r_{AB}} \sum_{i=1}^{d_A} \sum_{j=1}^{d_C} \ket{i}_A \otimes R_k \ket{i}_A \otimes \ket{j}_C \otimes S_k \ket{j}_C.
	\end{equation}
	By the argument above, this gives
	\begin{align}
		r_A &= \rk\left( \sum_{k=1}^{r_{AB}} R_k \otimes \ket{\theta_k}_{CD} \right),   \\
		r_{AC} &= \rk\left( \sum_{k=1}^{r_{AB}} R_k \otimes S_k \right) \label{rankAC}.
	\end{align}
	
	Now consider the following observation
	\begin{equation}
		\sum_{k=1}^{r_{AB}} R_k \otimes \ket{\theta_k}_{CD} = \left( \sum_{k=1}^{r_{AB}} R_k \otimes \1_C \otimes S_k \right) \left( \1_A \otimes \ket{\phi^+}_{CC} \right),
	\end{equation}
	which implies
	\begin{equation}
		r_A = \rk\left( \sum_{k=1}^{r_{AB}} R_k \otimes \ket{\theta_k}_{CD} \right) \leq \rk\left( \sum_{k=1}^{r_{AB}} R_k \otimes \1_C \otimes S_k \right) = r_{AC}\rk(\1_C).
	\end{equation}
	Notice that if we could replace $\rk(\1_C)$ by a term less than or equal to $r_{AB}$ then we would be done. This is what the following lemma allows us to do.
	
	\begin{lemma} \label{linalglemma}
		Suppose $S_1,\ldots,S_N$ are (non-zero) linear maps $S_i:\hilb{C}\to\hilb{D}$ such that $\trace(S_i^\dagger S_j)=0$ whenever $i\neq j$. Then for some $K\leq N$ there is a rank $K$ projector $P$ on $\hilb{C}$, and a state $\ket{\phi}_{CC}$, such that the vectors $\{(P\otimes S_k)\ket{\phi}_{CC}\}_{k=1}^N$ are linearly independent.
	\end{lemma}
	\begin{proof}
		Fix bases $\basis{C}=\{\ket{i}_C\}_{i=1}^{d_C}$ and $\basis{D}=\{\ket{j}_D\}_{j=1}^{d_D}$. Let $\ket{\psi_k}_{CD}:=(P\otimes S_k)\ket{\phi}_{CC}$, where $P:=\sum_{i=1}^K \ket{i}\bra{i}_C$, and $\ket{\phi}_{CC}:=\sum_{i=1}^K \ket{ii}_{CC}$. Then the $\ket{\psi_k}$ are linearly dependent if and only if there exist constants $\lambda_1,\ldots,\lambda_N$, not all zero, such that:
		\begin{equation}
		\begin{aligned}
			\sum_{k=1}^N \lambda_k \ket{\psi_k}_{CD} =0 \ 
			\iff &\ \sum_{k=1}^N \sum_{i=1}^K \ket{i}_C \otimes \lambda_k S_k \ket{i}_C = 0 \\
			\iff &\ \sum_{k=1}^N \lambda_k S_k \ket{i}_C =0 \quad \forall 1\leq i \leq K \\
			\iff &\ \text{The matrices formed by taking the first $K$ columns of each} \\
			     &\ \text{of $S_1,\ldots,S_N$ (when written with respect to $\basis{C}$ and $\basis{D}$)} \\
			     &\ \text{are linearly dependent.}
		\end{aligned}
		\end{equation}
	
		We proceed by induction. The lemma clearly holds for $N=1$. Suppose that it is true for $N=m$. Then there exists $\hat{K}\leq m$ such that the first $\hat{K}$ columns of the matrices $S_1,\ldots,S_m$ are linearly independent. Write $A^{(\hat{K})}$ for the first $\hat{K}$ columns of the matrix $A$. Consider matrix $S_{m+1}$. If $S_{m+1}^{(\hat{K})}$ is linearly independent of $S_1^{(\hat{K})},\ldots,S_m^{(\hat{K})}$ then we are done. Otherwise we have unique constants $\mu_1,\ldots,\mu_m$ such that $S_{m+1}^{(\hat{K})}=\sum_{k=1}^m \mu_k S_k^{(\hat{K})}$. Because the $S_k$ are orthogonal under the Hilbert-Schmidt inner product, they must be linearly independent, and so we cannot have $S_{m+1}=\sum_{k=1}^m \mu_k S_k$. Therefore, there must be a column in which equality does not hold. By relabelling basis vectors, we may assume that this is the $(\hat{K}+1)$-th column. Because of the uniqueness of the $\mu_k$ it follows that $S_1^{(\hat{K}+1)},\ldots,S_{m+1}^{(\hat{K}+1)}$ are linearly independent.
	\end{proof}
	
	In order to make use of this lemma, we first observe that $S_1,\ldots,S_{r_{AB}}$ are orthogonal under the Hilbert-Schmidt inner product. Indeed, since $\{ \ket{\theta_k}_{CD} \}_{k=1}^{r_{AB}}$ are mutually orthogonal states, for $i\neq j$ we have
	\begin{equation}
	\begin{aligned}
		0  = \braket{\theta_i}{\theta_j}
		  &= \sum_{k,l=1}^{d_C} (\bra{k}_C \otimes \bra{k}_C S_i^\dagger) 
		                                 (\ket{l}_C \otimes S_j \ket{l}_C) \\
		  &= \sum_{k,l=1}^{d_C} \delta_{kl} \bra{k}_C S_i^\dagger S_j \ket{l}_C 
%		  &= \sum_{k=1}^{d_C} \bra{k}_C S_i^\dagger S_j \ket{k}_C
		   = \trace S_i^\dagger S_j.
	\end{aligned}
	\end{equation}
	
	Let $K$, $P$ and $\ket{\phi}_{CC}$ be given by Lemma \ref{linalglemma}. 
	Then we can easily construct a linear map 
	$V_{CD}:\hilb{C}\otimes\hilb{D}\to\hilb{C}\otimes\hilb{D}$ which sends 
	$(P\otimes S_k)\ket{\phi}_{CC}$ to $(\1_C\otimes S_k)\ket{\phi^+}_{CC}$ for each 
	$k$. Then we can write the following:
	\begin{equation}
		\sum_{k=1}^{r_{AB}} R_k \otimes \ket{\theta_k}_{CD} 
		     = (\1_B\otimes V_{CD})\left(\sum_{k=1}^{r_{AB}} R_k \otimes P \otimes S_k\right)
		       (\1_A \otimes \ket{\phi}_{CC}).
	\end{equation}
	This implies
	\begin{equation}
	\begin{aligned}
		r_A &\leq \rk \left(\sum_{k=1}^{r_{AB}} R_k\otimes P \otimes S_k\right) \\
		    &= K \rk \left(\sum_{k=1}^{r_{AB}} R_k \otimes S_k\right) \\
	      	&\leq r_{AB} \, r_{AC},
	\end{aligned}
	\end{equation}
	and we are done.
\end{proof}

\section{Extremal states}
\label{sec:rays}
Having established two further inequalities for the ranks of multiparty states, 
it is natural to ask whether there are any more. At this point it is useful to 
switch to the perspective of the $0$-R\'enyi entropy, by taking the logarithm 
of the ranks. In this scenario, our inequalities become linear constraints on 
the set of possible entropy vectors.

Focusing on the four-party case, we have the following constraints:
\begin{equation}
\begin{aligned}
	S_0(A) &\geq 0, \\
	S_0(A) + S_0(B) &\geq S_0(AB), \\
	S_0(AB) + S_0(AC) &\geq S_0(A), \\
	S_0(AB) + S_0(AC) + S_0(BC) &\geq 2S_0(A),
\end{aligned}
\end{equation}
and all inequalities by permuting the names of the parties, and more generally
by substituting pairwise disjoint subsets of parties for $A$, $B$, $C$.

We can use symbolic computation software, for instance LRS \footnote{Available as a 
free download at {\tt http://cgm.cs.mcgill.ca/\~{ }avis/C/lrs.html}.} to find the 
extremal rays of the cone determined by these inequalities. We obtain eight
families of rays (up to permuting the parties) spanned by the vectors
listed below.

\medskip
\begin{center}
  {\begin{tabular}{c||c|c|c|c|c|c|c}
  Family & A & B & C & D & AB & AC & AD \\ \hline\hline
  1 & 1 & 1 & 0 & 0 & 0 & 1 & 1 \\ \hline
  2 & 1 & 1 & 1 & 1 & 2 & 2 & 2 \\ \hline
  3 & 1 & 1 & 1 & 1 & 1 & 1 & 0 \\ \hline
  4 & 2 & 2 & 1 & 1 & 2 & 1 & 1 \\ \hline
  5 & 2 & 2 & 2 & 1 & 2 & 1 & 1 \\ \hline
  6 & 3 & 3 & 3 & 1 & 2 & 2 & 2 \\ \hline
  7 & 2 & 2 & 2 & 1 & 3 & 1 & 1 \\ \hline
  8 & 1 & 1 & 1 & 1 & 2 & 1 & 0 \\ \hline
\end{tabular}}
\end{center}
\medskip

If we can find states with $0$-entropy vectors on, or arbitrarily close to these rays, then we can conclude that there are no more linear inequalities. Unfortunately, we were only able to do so for the first six families. Consider the following states:
\begin{equation}
\begin{aligned}
	\ket{\psi_1}_{ABCD} &= \left(\ket{00}_{AB}+\ket{11}_{AB}\right)\ket{00}_{CD}, \\
	\ket{\psi_2}_{ABCD} &= \sum_{i,j=0}^2 \ket{i}_A \ket{j}_B \ket{i+j}_C \ket{i+2j}_D,
\end{aligned}
\end{equation}
where, in the second state, addition takes place modulo $3$. These states have rank vectors:
\begin{equation}
\begin{aligned}
	r_{\psi_1}&=(2,2,1,1,1,2,2), \\
	r_{\psi_2}&=(3,3,3,3,9,9,9),
\end{aligned}
\end{equation}
from which it follows that their $0$-entropy vectors lie on rays 1 and 2. 
Now consider the states:
\begin{equation}
\begin{aligned}
	\ket{\psi_3}_{ABCD} &= \ket{\phi_d^+}_{AB} \oplus \ket{\phi_d^+}_{CD}, \\
	\ket{\psi_4}_{A_1A_2B_1B_2CD} &= \ket{\phi_d^+}_{A_1C}\ket{\phi_d^+}_{A_2D} \oplus \ket{\phi_d^+}_{B_1C} \ket{\phi_d^+}_{B_2D}, \\
	\ket{\psi_5}_{A_1A_2B_1B_2C_1C_2D} &= \ket{\phi_d^+}_{A_1C_1}\ket{\phi_d^+}_{A_2D} \oplus \ket{\phi_d^+}_{B_1C_1}\ket{\phi_d^+}_{B_2D} \oplus \ket{\phi_d^+}_{A_1C_1}\ket{\phi_d^+}_{B_1C_2}, \\
	\ket{\psi_6}_{A_1A_2A_3B_1B_2B_3C_1C_2C_3D} &= \ket{\phi_d^+}_{A_1B}\ket{\phi_d^+}_{A_2C}\ket{\phi_d^+}_{A_3D} \oplus \ket{\phi_d^+}_{AB_1}\ket{\phi_d^+}_{B_2C}\ket{\phi_d^+}_{B_3D} \\
		& \quad \oplus \ket{\phi_d^+}_{AC_1}\ket{\phi_d^+}_{BC_2}\ket{\phi_d^+}_{C_3D}.
\end{aligned}
\end{equation}
Here, we have used some notation that requires explanation. As before, the state 
$\ket{\phi_d^+}_{AB}$ is $\sum_{i=1}^d \ket{i}_A \ket{i}_B$. 
Where a system is missing from a state, it is assumed to be present, but in an 
unentangled pure state. Finally, the orthogonal sum of two states, denoted 
$\ket{\psi}\oplus\ket{\eta}$, is the superposition of $\ket{\psi}$ and $\ket{\eta}$ 
embedded in orthogonal parts of the local Hilbert spaces. For example,
\begin{equation}
\begin{aligned}
	\ket{\phi_d^+}_{AB} &= \ket{\phi_d^+}_{AB} \ket{11}_{CD} = \sum_{i=1}^d \ket{ii11}_{ABCD}, \\
	\ket{\phi_d^+}_{CD} &= \ket{11}_{AB} \ket{\phi_d^+}_{CD} = \sum_{i=1}^d \ket{11ii}_{ABCD},
\end{aligned}
\end{equation}
and we could have
\begin{equation}
	\ket{\psi_3}_{ABCD} = \sum_{i=1}^d \ket{ii11}_{ABCD} + \sum_{i=1}^d \ket{d+1}_A \ket{d+1}_B \ket{d+i}_C \ket{d+i}_D.
\end{equation}
Notice that if states $\ket{\psi}$ and $\ket{\eta}$ have rank vectors $r_\psi$ and $r_\eta$,
then $\ket{\psi}\oplus\ket{\eta}$ has rank vector $r_\psi + r_\eta$.
Using this fact, we can see that states $\ket{\psi_3},\ldots,\ket{\psi_6}$ have rank vectors:
\begin{equation}
\begin{aligned}
	r_{\psi_3} &= (d+1, d+1, d+1, d+1, 2, 2d, 2d), \\
	r_{\psi_4} &= (d^2+1, d^2+1, 2d, 2d, 2d^2, 2d, 2d), \\
	r_{\psi_5} &= (d^2+d+1, d^2+d+1, d^2+2d, 2d+1, 3d^2, 3d, 3d), \\
	r_{\psi_6} &= (d^3+2d, d^3+2d, d^3+2d, 3d, 3d^2, 3d^2, 3d^2).
\end{aligned}
\end{equation}
It is clear that for large $d$, the leading term in each component will dominate, 
and for the $0$-entropy (once we have taken the log) only the leading exponent matters.
More precisely, let $v_{\psi}$ be the $0$-entropy vector of state $\ket{\psi}$. 
Then, $\lim_{d \to \infty} \frac1{\log d} v_{\psi_i}$ is the $i$th vector in the table 
above, for each $i=3,4,5,6$.

For these four rays, we were unable to find states with $0$-entropy vectors actually on the ray. Interestingly, in the case of ray 3, it is easy to see that such a state cannot be found. This is because, for any state $\ket{\psi}$ on the ray, $S_0(AD)=0$, which implies that the rank across the $AD:BC$ partition is 1, so we can write $\ket{\psi}_{ABCD}=\ket{\eta}_{AD}\ket{\theta}_{BC}$. It follows that $r_{AB}=r_A r_B$ and so $S_0(AB)=S_0(A)+S_0(B)$, which does not hold for a state on ray 3. 
Similarly, one can show that it is not possible to find states on ray 6, but we do not know this 
for rays 4 and 5.

\begin{rmk}
	Notice that for large enough $d$, $\ket{\psi_3},\ket{\psi_4},\ket{\psi_5}$ and $\ket{\psi_6}$ all violate the inequality
	\begin{equation}
		r_Ar_Br_C \leq r_{AB}r_{AC}r_{AD},
	\end{equation}
	which had been conjectured in \cite{HdV13}. 
	Furthermore, all four provide further counterexamples to
	strong subadditivity.
\end{rmk}

\section{Another inequality?}
\label{sec:hypothesis}
In the previous section we took our known inequalities for the $0$-entropy, and computed the corresponding set of extremal rays. We then demonstrated that 6 of the 8 families of these extremal rays can be approximated by $0$-entropy vectors. However, we were unable to find such a construction for rays 7 and 8.

Now we apply the process in reverse. We have found 6 families of $0$-entropy vectors which are extremal rays. We can again use the LRS software to compute the set of inequalities to which they correspond. It turns out that this set is just the known inequalities, with one further inequality added, which we present as an hypothesis.
\begin{hyp} \label{rankhyp}
	Let $\ket{\psi}_{ABCD}$ be a four-party quantum state. Then $r_{BC}\leq r_{AB} \, r_{AC}$.
\end{hyp}
If we could prove this hypothesis, then we would have a complete picture of the four-party linear inequalities for the $0$-entropy.

Before presenting some partial results towards this hypothesis, we first write it in a different form.
\begin{hyp} \label{matrixhyp}
	Let $R_1,\ldots,R_K$ be $m_1 \times n_1$ complex matrices, and let $S_1,\ldots,S_K$ be $m_2 \times n_2$ complex matrices. Then
	\begin{equation}
		\rk\left( \sum_{i=1}^K R_i \otimes S_i^T \right) 
		       \leq K \rk\left( \sum_{i=1}^K R_i \otimes S_i \right).
	\end{equation}
\end{hyp}

\begin{lemma}
	Hypothesis \ref{rankhyp} is equivalent to Hypothesis \ref{matrixhyp}.
\end{lemma}
\begin{proof}
	\ref{matrixhyp}$\implies$\ref{rankhyp}: Let $\ket{\psi}_{ABCD}$ be a four-party quantum state, and let us write it in the same way as \eqref{RSdecomp}:
	\begin{equation}
		\ket{\psi}_{ABCD} = \sum_{k=1}^{r_{AB}} \sum_{i=1}^{d_A} \sum_{j=1}^{d_C} \ket{i}_A \otimes R_k \ket{i}_A \otimes \ket{j}_C \otimes S_k \ket{j}_C.
	\end{equation}
	Now, instead of considering $R_k, S_k$ as linear maps, we consider them as matrices written in the standard bases. From eq.~\eqref{rankAC} we have
	\begin{equation}
		r_{AC} = \rk\left( \sum_{k=1}^{r_{AB}} R_k \otimes S_k \right)
	\end{equation}
	Further, by using the identity 
	$(\1_C\otimes S_k)\ket{\phi^+}_{CC} = (S_k^T \otimes \1_C)\ket{\phi^+}_{DD}$, we can write
	\begin{equation}
		\ket{\psi}_{ABCD} = \sum_{k=1}^{r_{AB}} \sum_{i=1}^{d_A} \sum_{j=1}^{d_C} \ket{i}_A \otimes R_k \ket{i}_A \otimes S_k^T\ket{j}_D \otimes \ket{j}_D,
	\end{equation}
	from which it follows that
	\begin{equation}
		r_{AD} = \rk\left( \sum_{k=1}^{r_{AB}} R_k \otimes S_k^T \right).
	\end{equation}
	Hypothesis \ref{matrixhyp} then implies $r_{BC}\leq r_{AB} \, r_{AC}$ for $\ket{\psi}$.

	\ref{rankhyp}$\implies$\ref{matrixhyp}: Suppose that Hypothesis \ref{matrixhyp} is false. Then there exist matrices $R_k$, $S_k$ for $1\leq k \leq K$ such that
	\begin{equation}
		\rk\left( \sum_{i=1}^K R_i \otimes S_i^T \right) > K \rk\left( \sum_{i=1}^K R_i \otimes S_i \right).
	\end{equation}
	Consider the quantum state
	\begin{equation}
		\ket{\psi}_{ABCD} = \sum_{k=1}^{K} \sum_{i=1}^{d_A} \sum_{j=1}^{d_C} \ket{i}_A \otimes R_k \ket{i}_A \otimes \ket{j}_C \otimes S_k \ket{j}_C.
	\end{equation}
	For this state we have
	\begin{align}
		r_{AC} &= \rk\left( \sum_{k=1}^{K} R_k \otimes S_k \right), \\
		r_{AD} &= \rk\left( \sum_{k=1}^{K} R_k \otimes S_k^T \right),
	\end{align}
	and, since $r_{AB}=\dim\vecspan\{ (\1_C\otimes S_k)\ket{\phi^+}_{CC} : 1\leq k \leq K \}$, 
	we have $r_{AB}\leq K$. Putting this together we have a state $\ket{\psi}$ for which
	\begin{equation}
		r_{BC}=\rk\left( \sum_{k=1}^{K} R_k \otimes S_k^T \right) > K \rk\left( \sum_{i=1}^K R_i \otimes S_i \right) \geq r_{AB} \, r_{AC},
	\end{equation}
	so Hypothesis \ref{rankhyp} must also be false.
\end{proof}

\begin{thm}
	Let $\ket{\psi}_{ABCD}$ be a four-party quantum state, with $r_{AB}\leq 2$. 
	Then $r_{BC}\leq r_{AB} \, r_{AC}$.
\end{thm}
\begin{proof}
	By the argument above, it suffices to prove Hypothesis \ref{matrixhyp} for the cases $K=1$ and $K=2$.
	
	$K=1$: $\rk R_1\otimes S_1  = (\rk R_1)(\rk S_1) = (\rk R_1)(\rk S_1^T) = \rk R_1\otimes S_1^T$,
	using the fact that the row rank and the column rank of a matrix are the same, and equal to its rank.
	
	$K=2$: Let $M=R_1\otimes S_1 + R_2\otimes S_2$, where $R_1,R_2$ are $m_1\times n_1$ complex matrices, and $S_1,S_2$ are $m_2\times n_2$ complex matrices. This means we can write
	\begin{equation}
		M=\left( \begin{array}{c|c|c|c}
			M_{11} & M_{12} & \ldots & M_{1n_1} \\ \hline
			M_{21} & M_{22} & \ldots & M_{2n_1} \\ \hline
			\vdots & \vdots & \ddots & \vdots \\ \hline
			M_{m_11} & M_{m_12} & \ldots & M_{m_1n_1}
		\end{array} \right),
	\end{equation}
	where $M_{ij} = (R_1)_{ij} S_1 + (R_2)_{ij} S_2$. We refer to the matrices $M_{ij}$ as the `blocks' of $M$. Let $M^\Gamma$ denote the partial transpose of $M$, which is the matrix obtained by taking the transpose of each block of $M$. We aim to show that $\rk M^\Gamma \leq 2 \rk M$.
	
	Notice the following:
	\begin{enumerate}[(i)]
		\item Let $U:=\vecspan\{S_1,S_2\}$. Then $U$ is a 2-dimensional vector space of matrices, and $M_{ij}\in U$ for all $i,j$.
		\item Let $V:=\mathscr{C}(S_1 | S_2)$, where $\mathscr{C}(A)$ denotes the span of the columns of matrix $A$, and $S_1|S_2$ is a shorthand for the block matrix $(S_1 | S_2)$. Then every column of each $M_{ij}$ is in $V$.
		\item Suppose that $T_1,T_2$ are linearly independent blocks of $M$. Then clearly we have $\vecspan\{T_1,T_2\}=U$. Further, by applying elementary column operations we obtain
			\begin{equation}
			\begin{aligned}
				\mathscr{C}(T_1 | T_2) &= \mathscr{C}(T_1 | T_2 | 0 | 0) \\
					&= \mathscr{C}(T_1 | T_2 | S_1 | S_2) \\
					&= \mathscr{C}(0 | 0 | S_1 | S_2) \\
					&= \mathscr{C}(S_1 | S_2) \\
					&= V.
			\end{aligned}
			\end{equation}
		In particular, we must have $\dim V \leq \rk  T_1 + \rk T_2$,
		and hence $\rk T_i \geq \frac12 \dim V$ for at least one value of $i$.
		\item Each row of $M^\Gamma$ is in the space $V\oplus V\oplus \ldots \oplus V$ 
		(where there are $n_1$ copies of $V$) and so, for example, $\rk M^\Gamma \leq n_1 \dim V$.
	\end{enumerate}
	
	Suppose that $E$ is an $n_1 \times n_1$ invertible matrix, and let $\1$ be 
	the $n_2 \times n_2$ identity matrix. 
	Then, $\rk M(E\otimes \1) = \rk M$ and 
	$\rk [M(E\otimes \1)]^\Gamma = \rk M^\Gamma (E\otimes \1) = \rk M^\Gamma$, 
	so, if we choose to, we can replace $M$ with $M(E\otimes \1)$ for any invertible 
	$E$. In the case where $E$ performs elementary column operations, the effect of 
	$(E\otimes \1)$ acting on $M$ on the right is to perform block-wise column 
	operations on $M$. (By an analogous argument, we may also perform block-wise 
	row operations on $M$.)
	
	Suppose that $M_{11}$ and $M_{12}$ are linearly independent matrices. Then, for all 
	$j\geq 3$, $M_{1j}\in \vecspan\{M_{11}, M_{12}\}=U$, and by applying block-wise column 
	operations we may assume (as far as the ranks are concerned) that $M$ has the form
	\begin{equation}
		M=\left(\begin{array}{c|c|c|c|c}
		M_{11} & M_{12} & 0 & \ldots & 0 \\ \hline
		M_{21} & M_{22} & M_{23} & \ldots & M_{2n_1} \\ \hline
		\vdots & \vdots & \vdots & \ddots & \vdots \\ \hline
		M_{m_11} & M_{m_12} & M_{m_13} & \ldots & M_{m_1n_1}
		\end{array}\right).
	\end{equation}
	Let
	\begin{equation}
		M':=\left(\begin{array}{c|c|c}
		M_{23} & \ldots & M_{2n_1} \\ \hline
		\vdots & \ddots & \vdots \\ \hline
		M_{m_13} & \ldots & M_{m_1n_1}
		\end{array}\right),
	\end{equation}
	and consider the following inequalities, which hold for arbitrary block matrices:
	\begin{equation}
		\rk A + \rk C \leq \rk\left(\begin{array}{c|c} A & 0 \\ \hline B & C \end{array}\right) 
		              \leq \rk\left(\begin{array}{c} A \\ \hline B \end{array}\right) + \rk C.
	\end{equation}
	Applying the second inequality to $M^\Gamma$, and using the same argument as 
	comment (iv) above, we obtain
	\begin{equation}
		\rk M^\Gamma \leq 2 \dim V + \rk M'^\Gamma.
	\end{equation}
	Applying the first inequality to $M$ we obtain
	\begin{equation}
	\begin{aligned}
		\rk M &\geq \rk(M_{11} | M_{12}) + \rk M' \\
			  &=    \dim V + \rk M'.
	\end{aligned}
	\end{equation}
	Therefore, if $M$ is a counterexample to Hypothesis \ref{matrixhyp}, i.e. 
	$\rk M^\Gamma > 2\rk M$, then $M'$ is also a counterexample.
	
	Suppose, for contradiction, that a counterexample to Hypothesis \ref{matrixhyp} 
	exists, with $K=2$. Let $M$ be a minimal such counterexample, in the sense that 
	$m_1+n_1$ takes the smallest possible value. By the argument above we may assume 
	that $M_{11}$ and $M_{12}$ are not linearly independent. Further, by applying 
	block-wise row and column operations, we may assume that no two blocks of $M$ 
	in the same row are linearly independent.
	
	Now, $M$ must contain 2 linearly independent blocks, or else it would be 
	possible to write $M=R\otimes S$, which implies $\rk M^\Gamma =\rk M$. Suppose 
	$T_1,T_2$ are linearly independent blocks. By comment (iii) above, we may assume 
	that $\rk T_1 \geq\frac12 \dim V$. By swapping block-wise rows and columns, we 
	may also assume that $T_1$ is $M_{11}$. Since each row has only one linearly 
	independent block, using further block-wise column operations, we can assume 
	the matrix has the form
	\begin{equation}
		M=\left(\begin{array}{c|c|c|c}
		T_1 & 0 & \ldots & 0 \\ \hline
		M_{21} & M_{22} & \ldots & M_{2n_1} \\ \hline
		\vdots & \vdots & \ddots & \vdots \\ \hline
		M_{m_11} & M_{m_12} & \ldots & M_{m_1n_1}
		\end{array}\right).
	\end{equation}
	Now,
	\begin{equation}
		\rk M^\Gamma \leq \dim V + \rk M'^\Gamma,
	\end{equation}
	where $M'$ is the part of the matrix below the blocks of zeroes. Also,
	\begin{equation}
		\rk M \geq \rk T_1 + \rk M' \geq \frac12 \dim V + \rk M'.
	\end{equation}
	Since $\rk M^\Gamma  > 2\rk M$, this gives
	\begin{equation}
		\dim V + \rk M'^\Gamma > 2\left( \frac12 \dim V + \rk M' \right),
	\end{equation}
	and so
	\begin{equation}
		\rk M'^\Gamma > 2 \rk M',
	\end{equation}
	which is a contradiction to the assumption that $M$ was a minimal counterexample. 
	Hence there can be no counterexample and the theorem is proved.
\end{proof}

\section{Conclusions}
\label{sec:conclusion}
We have introduced the problem of determining the universal inequalities
between the ranks of partial traces of a general $n$-party pure state.
Specifically, taking logarithms, the linear inequalities between the
$0$-R\'enyi entropies became the object of our study, and we showed that
by using strong subadditivity for the von Neumann entropy we can derive
two new inequalities for the ranks.

The search for further inequalities lead, in the case of $n=4$ parties,
to several interesting families of states, and an hypothetical third
inequality, which if true, would result in a complete description of the
linear inequalities for the $0$-R\'enyi entropy. However, a proof or refutation
of this hypothesis remains the major open problem of our work.

Note that we were able to give purely algebraic proofs for all the inequalities
we found, except Theorem~\ref{newineq2}; to go beyond statements derived
from strong subadditivity, it seems fruitful to develop an algebraic proof
of Theorem~\ref{newineq2}, too, but this has eluded us so far.

Returning to the basic setup of section~\ref{sec:ranks}, we also note
that it is far from clear in which sense $\Omega_n$, the set of
$0$-entropy vectors, is best described by linear inequalities. Obviously,
$\Omega_n$ is not a cone since it is a discrete set (a certain subset
of $\log \mathbb{N}^{2^{n-1}-1}$), so this is a valid and important
question. Possible answers are suggested by the case $n=3$, where we
observed that the closed cone $\mathcal{C}_3$ generated by $\Omega_3$ has
the property that $\Omega_3 = \mathcal{C}_3 \cap \log \mathbb{N}^3$. The
analogue of this is not true for general $n$. However, could it be at least the
case that every log-integer point in the \emph{interior} of $\mathcal{C}_n$
is in $\Omega_n$? Or if not that, every interior log-integer point of 
sufficiently large norm? Any such statement would provide at least partial
justification for focusing on the linear inequalities, and we refer them
to the careful attention of the reader.

Finally, we close this discussion with a suggestion for further extension
of our theory: While here we considered the Schmidt rank only for all bipartitions,
tensor rank is a natural multi-party generalization (see~\cite{Strassen}
and \cite{BCS-book} for the general concept, and \cite{EB-multischmidt}
for its appearance in quantum information), so we might be tempted to
associate to each state vector $\ket{\psi}$ a vector of tensor ranks, one for
each partition of the ground set $[n]$ into arbitrarily many parts. Their
number is known as the \emph{Bell number}.
It is a very interesting, yet wide-open problem, to find the universal 
inequalities between the tensor ranks of the various partitions of the 
$n$ parties.

\section*{Acknowledgments}

The authors are indebted to their respective parents and grandparents,
great-grandparents and other antecedents, separate or common, for giving 
them the opportunity to walk the earth at this time. 
We also thank Mil\'an Mosonyi for useful conversations on the topic of this paper.

JC acknowledges support by the U.K. EPSRC.
MH acknowledges support by the EC Marie Curie fellowship ``QUACOCOS''.
NL and AW acknowledge support by the EC STREP ``QCS''.
AW furthermore acknowledges financial support by the ERC (Advanced Grant ``IRQUAT'')
and the Philip Leverhulme Trust.

\bibliography{References}
\bibliographystyle{unsrt}
\end{document}